\documentclass[letterpaper, 10 pt, conference]{ieeeconf}  % Comment this line out if you need a4paper

\IEEEoverridecommandlockouts                              % This command is only needed if 
\usepackage{mathrsfs}
\usepackage{graphicx}
\usepackage{wrapfig,lipsum}
\usepackage{pgfplots}
\usepackage{amsmath,amssymb,amsfonts}  
\usepackage{mathtools}
\usepackage{tikz, tikzscale}
\usepackage{caption}
\usepackage{subcaption}
\usepackage{enumerate}
\usepackage{comment}
\usepackage{booktabs}
\newtheorem{lem}{Lemma}
\newtheorem{proposition}{Proposition}
\newtheorem{defin}{Definition}
\newtheorem{ass}{Assumption}
\newtheorem{thm}{Theorem}

\newenvironment{lemma}{\begin{lem}}{\hfill  \end{lem}}
\newenvironment{definition}{\vskip 3pt\begin{defin}}{\vskip 3pt%% \hfill $\square$ 
\end{defin}}
\newenvironment{theorem}[1][\,\!]{\ignorespaces\begin{thm}[#1]\rm }{\hfill \mbox{\footnotesize$\square$} \end{thm}}
\newcommand{\Gs}{\mathcal{G}_s}

\renewcommand{\a}{\mathrm{a}}
\newcommand{\brho}{b_\rho}

\title{\LARGE \bf
Topology Identification of Dynamical Signed Graphs
}

\author{Pelin Sekercioglu \quad Nana Wang \quad Angela Fontan \quad Dimos V. Dimarogonas% <-this % stops a space
\thanks{*This work was supported in part by the Wallenberg AI, Autonomous Systems and Software Program (WASP) funded by the Knut and Alice Wallenberg (KAW) Foundation, the Horizon Europe EIC project SymAware (101070802), the ERC LEAFHOUND Project, the Swedish Research Council (VR), and Digital Futures. P. Sekercioglu, A. Fontan, and D. V. Dimarogonas are with the Department of Decision and Control Systems, KTH Royal Institute of Technology, SE-100 44 Stockholm, Sweden. email:
        {\tt\small \{pelinse,nanaw,angfon,dimos\}@kth.se}}%
}

\begin{document}

\maketitle
\thispagestyle{empty}
\pagestyle{empty}

%%%%%%%%%%%%%%%%%%%%%%%%%%%%%%%%%%%%%%%%%%%%%%%%%%%%%%%%%%%%%%%%%%%%%%%%%%%%%%%%
\begin{abstract}
We propose an adaptive control protocol for identifying the topology of dynamical networks interconnected over undirected graphs with cooperative and antagonistic interactions. The signed network is modeled using a repelling Laplacian. Topology identification relies on an edge-based formulation of the network and adaptive control protocols through the design of a persistently excited auxiliary network. Our approach guarantees the simultaneous identification and synchronization of the unknown signed network and establishes uniform semiglobal practical asymptotic stability of the estimation errors. Numerical simulations validate our theoretical results.
\end{abstract}

% \textcolor{blue}{\begin{keywords}
% Topology identification, signed networks, synchronization, multi-agent systems, repelling Laplacian.
% \end{keywords}}

%%%%%%%%%%%%%%%%%%%%%%%%%%%%%%%%%%%%%%%%%%%%%%%%%%%%%%%%%%%%%%%%%%%%%%%%%%%%%%%%
\section{INTRODUCTION}
Dynamical networks are used to represent a large number of real networks, \textit{e.g.,} social networks, biological networks, multi-robot systems, sensor systems \cite{newman2003structure}. For instance, the human immune system can be viewed as a complex dynamical network where different types of cells communicate to defend the body against pathogens. Similarly, in multi-robot systems, autonomous agents exchange information and adjust their behaviors based on local interactions to achieve complex tasks. The collective behavior emerging from these local interactions can be modeled using complex dynamical networks. The interaction topology plays a key role in control design, particularly for achieving synchronization, yet it may be partially or entirely unknown in practice.

In the literature on topology identification, existing approaches include prediction error-based \cite{weerts2018prediction}, adaptive control-based \cite{zhu_new_2021}, constrained Lyapunov equation-based \cite{van_waarde_topology_2019}, causality-based \cite{li2023topology} methods, and learning-based methods \cite{11045067}. In this work, we focus on adaptive control-based methods for topology identification, where the control of the original network is designed to track an auxiliary network. Under certain linear independence conditions, these methods guarantee correct topology identification. However, once the network achieves synchronization, such conditions become difficult to verify or maintain. To overcome this, works in \cite{zhu_new_2021, wang2023finite, restrepo2023simultaneous,restrepo2023simultaneousTCNS} introduce external inputs to prevent synchronization and provide the persistency of excitation for the topology estimation.

In all the references above, agents are assumed to cooperate with each other. However, many real-world scenarios involve antagonistic behaviors, such as in robotics applications \cite{IEEE-robotmanip} and social networks where agents compete \cite{Fontan2021Signed}. A common framework for modeling both cooperation and antagonism in dynamical systems is that of \textit{signed graphs} \cite{altafini2012_6329411,Shi2019Dynamics}. In the current literature, there are two different definitions of signed Laplacians for signed graphs, which depend on how the diagonal elements are calculated from the weighted adjacency matrix. They are called \textit{opposing} and \textit{repelling} signed Laplacians in \cite{Shi2019Dynamics}. In this paper, we focus on \textit{repelling} signed Laplacian, whose main feature is that it always has 0 as an eigenvalue, but may or may not be stable since it can have diagonal negative entries. Conditions for the stability of the repelling Laplacian can be found in \cite{9965602}.

We propose an approach for undirected signed network topology identification, based on the \textit{edge-based formulation} and \textit{repelling signed Laplacians}. We build upon the framework proposed in \cite{restrepo2023simultaneous,restrepo2023simultaneousTCNS}, which transforms the synchronization problem into a stabilization problem of the origin in the error coordinates. Unlike these prior works, which assume \textit{all cooperative} interactions, we consider networks with signed topologies, where agents may also exhibit \textit{antagonistic} interactions. Our topology identification algorithm follows an adaptive control approach using the concept of $\delta$-persistency of excitation \cite{panteley2001relaxed}, and we establish uniform semi-global practical asymptotic stability of the origin in terms of the signed topology estimation errors. 
%In contrast to \cite{restrepo2023simultaneous,restrepo2023simultaneousTCNS}, the presence of \textit{antagonistic} interactions implies that the associated Laplacian matrix may have negative eigenvalues and may be unstable. 
Nevertheless, our method enables simultaneous topology identification and synchronization of the network, even when the overall system is initially unstable, provided some global information about the network is available, \textit{e.g.,} the smallest eigenvalue of the repelling signed Laplacian. We then relax this requirement by proposing a graph-theoretical approach for identifying the signed network topologies with edge weights of magnitudes less than or equal to one, allowing each agent to compute its control input in a fully distributed manner, without prior global knowledge.

\textit{Notation:} $\lvert \cdot \rvert$ denotes the absolute value for scalars, the Euclidean norm for vectors, and the spectral norm for matrices. $\mbox{diag}(z)$ denotes a diagonal matrix whose diagonal elements are the entries of the vector $z$. $\mathbb{R}$ is the set of real numbers and $\mathbb{R}_{\geq 0}$ the nonnegative orthant. $A>0$ ($A\ge 0$) indicates that $A$ is a positive definite (positive semidefinite) matrix. A continuous function $\alpha:\mathbb{R}_{\geq 0}\to \mathbb{R}_{\geq 0} \in \mathcal{K}$ if it is strictly increasing and $\alpha(0)=0$; $\alpha \in \mathcal{K}_{\infty}$ if, in addition, $\alpha(s)\to \infty$ as $s \to \infty$. We denote by $\mathbb{B}(\delta) \subset \mathbb{R}^n$ the open ball of radius $\delta$ centered at the origin, \textit{i.e.,} $\mathbb{B}(\delta) := \{ x \in \mathbb{R}^n : \lvert x \rvert < \delta\}$. We use the notation $\mathbb{H}(\delta,\Delta) := \{ x \in \mathbb{R}^n : \delta \leq |x| \leq \Delta\}$. Let $\Gs = (\mathcal{V}, \mathcal{E})$ denote a signed graph, where $\mathcal{V} = \{ \nu_1, \nu_2, \dots, \nu_N \}$ is the set of nodes and $\mathcal{E} \subseteq \mathcal{V} \times \mathcal{V}$ is the set of $M$ edges. The adjacency matrix of \( \Gs \) is \( A := [a_{ij}] \in \mathbb{R}^{N \times N} \), where \( a_{ij} \neq 0 \) if and only if \( (\nu_j, \nu_i) \in \mathcal{E} \). Each edge $(\nu_j, \nu_i)$ has either a positive sign (\( a_{ij} > 0 \)), representing a cooperative relationship, or a negative sign (\( a_{ij} < 0 \)), representing an antagonistic relationship. If all edges have a positive sign, the graph is called an \textit{unsigned} graph. The graph is \textit{undirected} if the information flow between agents is bidirectional, meaning $( \nu_i, \nu_j ) = ( \nu_j, \nu_i ) \in \mathcal{E}$. Otherwise, the graph is \textit{directed}. Self-loops are not considered. The degree matrix of \( \Gs \) is \( D := \mbox{diag}\{d_{ii}\} \in \mathbb{R}^{M \times M} \), where \( d_{ii} = \sum_{j=1}^N a_{ij} \). A spanning tree of a graph is a subgraph that includes all the vertices of the original graph and is a tree, meaning it is connected and has no cycles.

\section{Preliminaries on Persistency of Excitation}
Before delving into details, we recall the following crucial definitions from \cite{panteley2001relaxed} on ($\delta$-)persistency of excitation and lemmata from \cite{panteley2001relaxed,loria2005nested,loria1999new}, which serve as the basis for our stability analysis of the estimation errors. 

Consider the system 
\begin{equation}\label{104}
    \dot x = f(t,x),
\end{equation}
where $f: \mathbb{R}_{\geq 0} \times \mathbb{R}^{n_1+n_2} \to \mathbb{R}^{n_1+n_2}$ such that all solutions $x(t,t_0,x_0)$ with $x(0)=x_0$ are defined for all $t\ge t_0$. Let $x \in \mathbb{R}^n$ be partitioned as $x^\top := [x_1^\top \ x_2^\top]$, where $x_1 \in \mathbb{R}^{n_1}$ and $x_2 \in \mathbb{R}^{n_2}$.

\begin{definition} (\cite{panteley2001relaxed})\label{PE} $\phi(\cdot)$ is bounded and globally Lipschitz. The pair $(\phi,f)$ is called \textit{uniformly persistently exciting} (u-PE) if, for each $r>0$, there exist $\mu, T>0$ such that for all $(t_0, x_0) \in \mathbb{R}_{\ge0} \times \mathbb{B}(r)$, all solutions of \eqref{104} satisfy $\int_t^{t+T}\phi(\tau,x(\tau,t_0,x_0))\phi(\tau,x(\tau,t_0,x_0))^{\top}d\tau \ge \mu I, \forall t\geq t_0$.
\end{definition}
\begin{definition} (\cite{panteley2001relaxed})\label{delta-PE} The pair $(\phi,f)$ is called \textit{uniformly $\delta$-persistently exciting} (u$\delta$-PE) with respect to $x_1$ if, for each $r$ and $\delta>0$, there exist constants $T(r,\delta)>0$ and $\mu(r,\delta)>0$, such that for all $(t_0, x_0) \in \mathbb{R}_{\ge0} \times \mathbb{B}(r)$, all solutions of \eqref{104} satisfy $\min_{s\in [t, t+T]} |x_1(s)| \geq \delta \Rightarrow$
$\int_t^{t+T}\phi(\tau,x(\tau,t_0,x_0))\phi(\tau,x(\tau,t_0,x_0))^{\top}d\tau \ge \mu I, \forall t\geq t_0$.
\end{definition}
With an abuse of notation, we call the function $\phi$, u$\delta$-PE with respect to $x_1$ if the pair $(\phi,f)$ is u$\delta$-PE with respect to $x_1$. Define the set $\mathcal{D} := (\mathbb{R}^{n_1} \setminus \{ 0\}) \times \mathbb{R}^{n_2}$. Then, we have the following.% lemma from \cite{loria2005nested}.
\begin{lemma}(\cite[Lemma 1]{loria2005nested})\label{lemma2} If $x \mapsto \phi(t,x)$ is continuous uniformly in $t$, then $\phi$ is u$\delta$-PE with respect to $x_1$ if and only if for each $x \in \mathcal{D}$, there exist $T>0$ and $\mu >0$ such that $\int_t^{t+T}|\phi(\tau,x)|d\tau \geq \mu$ for all $t \in \mathbb{R}$.
\end{lemma}

The following lemma establishes that the output of a strictly proper, minimum phase and strictly stable filter, driven by a u$\delta$-PE input signal, is also u$\delta$-PE. A reminiscent version of the Lemma is originally presented in \cite{panteley2001relaxed}.
\begin{lemma} \label{filprop}
    Let $\phi : \mathbb{R}_{\ge 0} \times \mathbb{R}^n \to \mathbb{R}^m$ and consider
\begin{align}\label{eq:filtration_system}
\begin{bmatrix}
\dot{x}\\
\dot{\omega}
\end{bmatrix}
=
\begin{bmatrix}
f(t, x, \omega) \\
f_1(t, \omega) + f_2(t, x)\omega + \phi(t, x)
\end{bmatrix}
\end{align}
with $f_1 : \mathbb{R}_{\ge 0} \times \mathbb{R}^n \to \mathbb{R}^m$ Lipschitz in $\omega$ uniformly in $t$ and measurable in $t$, and satisfying $|f_1(\cdot)| \le |\omega|$ for all $t$;
$f_2 : \mathbb{R}_{\ge 0} \times \mathbb{R}^n \to \mathbb{R}^{m \times m}$ is locally Lipschitz in $x$ uniformly in $t$ and measurable in $t$.
Assume that $\phi(t, x)$ is u$\delta$-PE with respect to $x$. Assume also that $\phi$ is locally Lipschitz and there exists a non-decreasing function $\alpha : \mathbb{R}_{\ge 0} \to \mathbb{R}_{\ge 0}$ such that, for almost all $(t, x) \in \mathbb{R}_{\ge 0} \times \mathbb{R}^n$:
\begin{equation}
\max \left\{
|\phi(\cdot)|,\, |f(\cdot)|,\, |f_2(\cdot)|,\,
\left| \frac{\partial \phi(\cdot)}{\partial t} \right|,\,
\left| \frac{\partial \phi(\cdot)}{\partial x} \right|
\right\}
\le \alpha(|x|).
\label{eq:filtration_alpha}
\end{equation}
Assume further that all solutions $t \mapsto x_\phi$ of \eqref{eq:filtration_system}, with $x_\phi^{\top} = [x^{\top} \; \omega^{\top}]$, are defined in $[t_0, \infty)$ and satisfy
$|x_\phi(t, t_0, x_{\phi_0})| \le r, \forall t \ge t_0,$ then $\omega$ is u$\delta$-PE with respect to $x$.
\end{lemma}
Finally, the following lemma from \cite{loria1999new} will be used in the subsequent analysis.
%our main results on the stability analysis (in Propositions \ref{prop1} and \ref{prop2}).
\begin{lemma} \label{kyplemma}
For the system 
\begin{align}\label{lemma2eq}
\begin{bmatrix}
		\dot{{x}}_1 \\ \dot{x}_2 
	\end{bmatrix} = 
	\begin{bmatrix}
		A&
			B \psi(t,x(t))^\top \\
			-\psi(t,x(t))C^\top & 0
	\end{bmatrix} 	\begin{bmatrix}
		x_1 \\ x_2
	\end{bmatrix}\end{align}
suppose that there exist continuous non-decreasing functions 
$\alpha_i : \mathbb{R}_{\ge 0} \to \mathbb{R}_{\ge 0}$, $i = 1, 2$, such that $|\psi(t, x)| \le \alpha_1(|x|),$
$$\max \left\{
\left| \frac{\partial \psi(t, x)}{\partial t} \right|,
\left| \frac{\partial \psi(t, x)}{\partial \xi} \right|
\right\} \le \alpha_2(|x|),$$
and assume that $\psi(t, x(t))$ is u$\delta$-PE with respect to $x$.
If the triple $(A, B, C)$ satisfies the Kalman--Yakubovich--Popov lemma, i.e., there exist $P = P^{\top} > 0$ and $Q = Q^{\top} > 0$ such that
$A^{\top}P + PA = -Q$ and $PB = C$, then the origin of \eqref{lemma2eq} is uniformly globally
asymptotically stable.
\end{lemma}

\section{Model and Problem Formulation}\label{sec:m-p-formulation}
Consider a group of $N$ agents interconnected over an unknown and undirected topology described by an undirected signed graph $\Gs$. Each agent is modeled by 
\begin{equation}\label{CN}
	\dot{x}_i = f_i(x_i) - \sum_{j=1}^N a_{ij} (x_i - x_j) + u_i
\end{equation}
where $i,j \in \{ 1,2, \dots, N\}$, $x_i \in \mathbb{R}$ is the state of agent $i$, $f_i : \mathbb{R} \to \mathbb{R}$ is a smooth function, $a_{ij}$ is the \textit{unknown} adjacency weight of the interconnection between agents $i$ and $j$. More precisely, $a_{ij}=0$ if agents $i$ and $j$ are not interconnected, and $a_{ij}\neq0$ otherwise, where $a_{ij}>0$ if agents $i$ and $j$ are cooperative, and $a_{ij}<0$ if they are antagonistic.

Defining $x := [x_1\ x_2\ \dots \ x_N]^\top$, $F(x) := [f(x_1)\ f(x_2)\ \dots \ f(x_N)]^\top$ and $u:= [u_1\ u_2\ \dots \ u_N]^\top$, the system in \eqref{CN} may be expressed in compact form as
\begin{align}\label{vec_L}
	\dot{x} = F(x) - L x + u,
\end{align}
where $L \in \mathbb{R}^{N \times N}$ denotes the \textit{repelling} signed Laplacian matrix (hereafter simply denoted signed Laplacian, or by the symbol $L$, to simplify the notation) of the signed network $\Gs$, and its elements are defined by ${\ell_{ij}} = \sum\limits_{k =i}^N {{ a_{ik} }}$ if $i = j$ and ${\ell_{ij}} = - {a_{ij}}$ if $i \ne j$ \cite{Shi2019Dynamics}. For an undirected signed graph, the signed Laplacian $L$ is symmetric and all its eigenvalues are real. Moreover, it has at least one zero eigenvalue, and $\mathbf{1}\in \mathbb{R}^N$ is always an eigenvector of $L$ associated with the zero eigenvalue. However, in contrast to the \textit{conventional} Laplacian matrix (associated with a unsigned graphs, i.e., a network with all cooperative interactions), $-L$ is not Metzler, may not be diagonally dominant\footnote{A matrix is Metzler if all its off-diagonal elements are nonnegative; it is diagonally dominant if the absolute value of each diagonal element is greater than or equal to the sum of the absolute values of all the other (off-diagonal) elements in that same row.}, and may be indefinite, meaning it can have both positive and negative eigenvalues \cite{9965602}. This in turn implies that the (uncontrolled) system \eqref{vec_L} may have unstable dynamics. See examples in \cite{9965602} for unstable multi-agent systems described by repelling signed Laplacian dynamics.

The signed Laplacian matrix $L$ of the undirected signed network $\Gs$ can be expressed by
\begin{align}\label{453}
	L = E W_s E^\top,
\end{align}
where $W_s := \mbox{diag}\{ w_{s_k} \} \in \mathbb{R}^{M \times M}$, $M$ is the number of edges, and $w_{s_k}$ corresponds to the \textit{unknown} weight of the signed edge $\varepsilon_k = (i,j)$ interconnecting agents $i$ and $j$, which can be either positive or negative. In \eqref{453}, $E$ is the incidence matrix describing the interaction topology of the network. It is defined by first assigning an arbitrary orientation to each edge, after which its entries are given by ${[E]_{ik}} :=+1,$ if $\nu_i$ is the initial node of the edge $\varepsilon_k$; ${[E]_{ik}} := -1,$ if $\nu_i$ is the terminal node of the edge $\varepsilon_k$; ${[E]_{ik}} := 0$, otherwise.
% \begin{equation}\label{def:E}
% 	{[E]_{ik}} :=  \left\{ \begin{matrix*}[l]
% 		{+1,}&{\text{if}\ v_i\ \text{is the initial node of the edge}\ \varepsilon_k ;}\\
% 		{-1,}&{\text{if } v_i\ \text{is the terminal node of the edge}\ \varepsilon_k ;}\\
% 		{0,}&{\text{otherwise}}.
% 	\end{matrix*}\right.
% \end{equation}

To address the problem of topology identification, following the framework in \cite{restrepo2023simultaneous,restrepo2023simultaneousTCNS}, we first define a complete signed graph $\bar{\Gs}(\mathcal{V},\bar{\mathcal{E}})$, where $\bar{\mathcal{E}} \subseteq \mathcal{E}$ is the set of $\bar M:= \frac{1}{2}N(N-1)$ signed edges, and we denote by $\bar E \in \mathbb{R}^{N \times \bar M}$ its incidence matrix. Note that $\bar E$ is known, and can be constructed, since it corresponds to a complete signed graph with a known number of nodes, equivalent to the number of agents in the system \eqref{CN}.
We then introduce the signed weight matrix $\bar{W}_s := \mbox{diag}\{ \bar{w}_{s_k} \} \in \mathbb{R}^{\bar M \times \bar M}$. Here, $\bar{w}_{s_k} = {w}_{s_k} \ne 0$ if the interaction exists in the true signed graph $\Gs$, \textit{i.e.,} $\bar{\varepsilon}_k \in \mathcal{E}$ and $\bar{w}_{s_k} = 0$ otherwise, \textit{i.e.,} $\bar{\varepsilon}_k \in \bar{\mathcal{E}} \setminus \mathcal{E}$. 
This formulation allows the signed Laplacian matrix of the true signed graph $\Gs$ to be represented as $L = \bar E \bar W_s \bar E^\top$.
Therefore, \eqref{CN} can be expressed, in vector form, as 
\begin{align}\label{vec_E_multi}
	\dot{x} = F(x) - \bar E \bar W_s \bar E^\top x + u.
\end{align}
Consequently, the topology identification problem can be expressed as the estimation problem of the entries of the diagonal signed weight matrix $\bar{W}_s$, \textit{i.e.}, the weights $\bar{w}_{s_k}$.

\section{Main Results}
The control objective is to design the external input $u$ in order to simultaneously identify the unknown weights $\bar{w}_{s_k}$ on the signed graph $\Gs$ and synchronize the dynamical systems \eqref{CN}, even if, as discussed in Section~\ref{sec:m-p-formulation}, the system may have unstable dynamics. As a first step to achieve synchronization and topology identification, we design an auxiliary variable $\hat{x}$ and ensure that the network system \eqref{vec_L} converges to $\hat{x}$. By designing $\hat{x}$ to be persistently exciting (PE) (see Definition \ref{PE}), we establish uniform global asymptotic stability of the origin (Proposition \ref{prop1}). Next, we redesign the control gain using a game-theoretical approach (Proposition \ref{prop2}) independent of global knowledge of the network. Finally, we show uniform semiglobal practical asymptotic stability of the origin (Proposition \ref{prop4}).

We first make sure that the network system \eqref{vec_E_multi} converges to an auxiliary variable $\hat{x}$. Let the external input be set to
\begin{align}\label{u}
	u = -F(x) -c_1(x - \hat x(t)) + \dot{\hat{x}}(t) + \hat L \hat x(t),
\end{align}
where $\hat{x}(t)$ will be defined later. Here, $c_1$ is a positive constant and $\hat L$ is the estimated signed Laplacian containing the estimated weights of the signed graph defined as $\hat L = \bar E \hat W_s \bar E^\top,$ where $\hat{W}_s := \mbox{diag}\{ \hat{w}_{s_k} \} \in \mathbb{R}^{\bar M \times \bar M}$ are the estimated weights and $\bar M= \frac{1}{2}N(N-1)$. The update law for the weights is given as 
\begin{align}\label{updatelaw_two}
	\dot{\hat w}_s &= - \mbox{diag} \{ \bar E^\top \hat x \} \bar E^\top \tilde x,
\end{align}
where $\tilde x := x - \hat{x}(t)$ indicates the state synchronization error.

We define the weight estimation errors as $\tilde w_s := \bar w_s - \hat w _s$, let $\xi^\top = [\tilde x^\top \ \tilde{w}_s ^\top]$. We set $\hat{x}$ to be the state of an auxiliary dynamical system, whose update law is given as follows.
\begin{align}\label{eq2}
	\dot{\hat{x}}(t) = -\hat{x}(t) + \phi(t, \xi),
\end{align}
where $\phi(t, \xi): \mathbb{R}_{\geq 0} \times \mathbb{R}^{2 \bar M} \to \mathbb{R}^{\bar M}$, $(t,\xi)\mapsto \phi(t,\xi)$ is a function to be defined later. For this function, we pose the following assumption.
\begin{ass}\label{assumption1}
Let $\rho:\mathbb{R}_{\geq 0}\to \mathbb{R}_{\geq 0}$ be a continuous non-decreasing function. For all $\xi \in \mathbb{R}^{\bar M}$ and almost all $t \in \mathbb{R}_{\geq 0}$
\begin{align}\label{ass1}
	\max \left\{ \lvert \phi(\cdot) \rvert, \lvert \frac{\partial \phi(\cdot)}{\partial t} \rvert , \lvert \frac{\partial \phi(\cdot)}{\partial \xi} \rvert \right\} \leq \rho(\lvert \xi \rvert).
\end{align}
\end{ass}

From \eqref{vec_E_multi}, \eqref{u} and \eqref{eq2}, the dynamics of the error coordinates is given by 
\begin{align*} 
	\dot{\tilde{x}} &= \dot{x} - \dot{\hat{x}}= -( c_1 I + L) \tilde x(t) - \tilde L \hat x (t)\\
	%&= -( c_1 I + L) \tilde x(t) - \bar E \hat Z (t) \tilde w_s(t) \\
	\dot{\tilde{w}}_s &= \dot{\bar{w}}_s - \dot{\hat{{w}}}_s= \mbox{diag}\{\bar E^\top \hat x(t) \} \bar E^\top \tilde x(t) 
\end{align*}
where $\dot{\bar{w}}_s=0$ since ${\bar{w}}_s$ is constant and $\tilde L := \bar E \tilde W_s(t) \bar E^\top$. Equivalently, the closed-loop system of the error dynamics is given by, in vector form,
\begin{align}\label{eq1}
	\begin{bmatrix}
		\dot{\tilde{x}} \\ \dot{\tilde{w}}_s 
	\end{bmatrix} = 
	\begin{bmatrix}
		-c_1 I - L&
			-\bar E \hat{Z}(t) \\
			\hat{Z}(t) \bar E^\top & \mathbf{0}_{\bar M \times \bar M}
	\end{bmatrix} 	\begin{bmatrix}
		\tilde{x} \\ \tilde{w}_s
	\end{bmatrix},
\end{align}
where $\hat{Z}(t) = \mbox{diag}\{ \hat{z}(t) \}$ with $\hat{z}(t) = \bar E^\top \hat x(t)$.

In the literature on topology identification (see, e.g., \cite{zhu_new_2021}), systems of the form \eqref{eq1} have been proven to be globally exponentially stable if $\hat x(t)$ is designed to be PE (see Definition \ref{PE}), a property that directly follows from Assumption~\ref{assumption1} and the definition of $\hat x(t)$ in \eqref{eq2}.  
The corresponding result is presented in Proposition~\ref{prop1}.
 However, even if this allows for topology identification, the synchronization of the system cannot be guaranteed (see \cite{zhu_new_2021}). Then, in order to simultaneously identify the topology and synchronize the system \eqref{vec_E_multi}, we design $\phi(t, \xi)$ to be u$\delta$-PE as per Definition \ref{delta-PE}.

\begin{proposition}\label{prop1} Let $\phi: \mathbb{R}_{\geq 0} \times \mathbb{R}^{2 \bar M} \to \mathbb{R}^{\bar M},$ $(t, \xi) \mapsto \phi(t, \xi)$ satisfying \eqref{ass1} and $\lambda_{\min}(L)$ be the smallest eigenvalue of $L$. Then, for any $c_1 > -\lambda_{\min}(L)$, if $\phi$ is u$\delta$-PE, the origin of the closed-loop system \eqref{eq1} with $\hat x(t)$ given by the update law \eqref{eq2} is uniformly globally asymptotically stable and the topology of network \eqref{vec_E_multi} is identified by the estimate $\hat w$, that is, $\bar w = \lim_{t \to \infty}\hat w (t)$ with the update law \eqref{updatelaw_two}. Moreover, the system achieves synchronization.
\end{proposition}
\begin{proof}
Consider the following Lyapunov function
\begin{equation}\label{V1}
    V_1(\tilde x, \tilde w_s) := \frac{1}{2} \tilde{x}^\top\tilde{x} + \frac{1}{2} \tilde{w}_s^\top \tilde{w}_s.
\end{equation}

Its derivative along the trajectories of \eqref{eq1} yields
\begin{align}
\label{dotV1}
	\dot{V}_1 (\tilde x, \tilde w) = &- \tilde{x}^\top (c_1 I + L) \tilde x - \tilde{x}^\top \bar E \hat{Z}_s(t) \tilde{w}_s \nonumber \\
    &+ \tilde{w}_s^\top \hat{Z}_s(t) E_s^\top \tilde x = - \tilde{x}^\top (c_1 I + L ) \tilde x.
\end{align}

Notice that here, if $-L$ is marginally stable\footnote{$-L$ is called marginally stable (of corank $1$) if $0$ is a unique eigenvalue of $-L$ and all its other eigenvalues are negative (see \cite{9965602} for details).}, then $c_1 I + L > 0$ for any $c_1 > 0$. On the other hand, if $-L$ has negative eigenvalues, from the assumption in Proposition~\ref{prop1} that $c_1 > - \lambda_{\min}(L)$, we can conclude that the system \eqref{eq1} is uniformly globally stable, and its solutions $\xi(t,t_0,\xi_0)$ are uniformly globally bounded. 

Then, for the auxiliary system \eqref{eq2}, we consider the Lyapunov function $\hat V (\hat x) = \frac{1}{2} \lvert \hat x \rvert^2$, and its derivative satisfies $\dot{\hat V} (\hat x) \leq - \lvert \hat x \rvert^2 + \lvert \hat x \rvert \lvert \phi(t, \xi) \rvert \leq  - \lvert \hat x \rvert^2 + \lvert \hat x \rvert \lvert \rho(\lvert \xi \rvert) \rvert$, which we obtained from \eqref{ass1}. Let $\epsilon>0$ be such that $c_1':= 1-\frac{\epsilon}{2} >0$. Using Young's inequality on the second term, we have $\dot{\hat V} (\hat x)  \leq - c_1' \lvert \hat x \rvert^2 + \frac{1}{2\epsilon}\lvert \rho(\lvert \xi \rvert) \rvert^2$. Since the solutions $\xi(t,t_0,\xi_0)$ are uniformly globally bounded, there exists $c_2>0$ such that $\lvert \rho(\lvert \xi(t,t_0,\xi_0) \rvert) \rvert^2 \leq c_2$ for all $t$. Then,
\begin{align}
	\dot{\hat V} (\hat x) & \leq - c_1' \lvert \hat x \rvert^2 + c_2. \label{abc}
\end{align}
We conclude that the solutions $\hat x (t,t_0,\hat x_0)$ are uniformly globally bounded. Therefore, since all the assumptions in Lemma \ref{filprop} are satisfied, $\hat x (t)$ is $u \delta$-PE with respect to $\xi$. The rest of the proof follows from Lemma \ref{kyplemma} by replacing $x_1$, $x_2$, $A$, $B$, and $\psi(t,x(t))$ with $\tilde x$, $\tilde{w}_s$, $-c_1 I - L$, $\bar E$ and $\hat Z(t)$, respectively. Thus, we conclude that the origin of \eqref{eq1} is uniformly globally asymptotically stable. Therefore, we have $\lim_{t \to \infty} x(t) = \hat x(t)$ and $\lim_{t \to \infty}\hat w_s = \bar w_s$. Moreover, from \eqref{abc}, we have $\lim_{t \to \infty} \hat x(t) = 0$ and $\lim_{t \to \infty} x(t) = 0$.
%since $\left( (-c_1 I - L), \bar E, \bar E \right)$ satisfies the Kalman-Yakubovich-Popov (KYP) lemma, \textit{i.e.,} for $c_1 > -\lambda_{\min}$, there exist $P = P^\top$ and $Q=Q^\top$ such that $(-c_1 I - L)^\top P + P (-c_1 I - L) = -Q$ and $P \bar E = \bar E$ with $P=I$, from Lemma~\ref{kyplemma} in the Appendix, the origin of \eqref{eq1} is uniformly globally asymptotically stable. Therefore, we have that $\lim_{t \to \infty} x(t) = \hat x(t)$ and $\lim_{t \to \infty}\hat w_s = \bar w_s$. Moreover, from \eqref{abc}, we have that $\lim_{t \to \infty} \hat x(t) = 0$ and $\lim_{t \to \infty} x(t) = 0$.
\end{proof}

Proposition~\ref{prop1} relies on the knowledge of a global parameter of the network, which is the smallest eigenvalue of the signed Laplacian matrix. Indeed, in the case of a signed graph, the signed Laplacian $L$ can be indefinite, meaning that its smallest eigenvalue may be negative, i.e., $\lambda_{\min}(L)<0$. Intuitively, this situation arises when the number of negative edges in the signed graph is large (see \cite{fontan2023pseudoinverses,9965602} for details and examples). 
Therefore, to stabilize the system, the parameter $c_1$ in Proposition~\ref{prop1} needs to be chosen big enough to counteract the effect of this negative smallest eigenvalue, which requires some prior knowledge on the network (in this case, of $\lambda_{\min}(L)$).
To address this limitation, we present a result, independent of global knowledge, for signed networks having interconnection weights between $1$ and $-1$ (Proposition~\ref{prop2}). For that, we pose the following assumption.
\begin{ass}\label{assumption2}
The signed graph satisfies $|\bar w_{s_k}| \leq 1$, where $\bar w_{s_k}$ are \textit{unknown} (but normalized) weights.   
\end{ass}
\begin{proposition}\label{prop2} 
%{\color{red} Assume that the signed topology satisfies $|\bar w_{s_k}| \leq 1$, where $\bar w_{s_k}$ are \textit{unknown} (but normalized) weights.}
Let $\phi: \mathbb{R}_{\geq 0} \times \mathbb{R}^{2 \bar M} \to \mathbb{R}^{\bar M},$ $(t, \xi) \mapsto \phi(t, \xi)$ satisfying \eqref{ass1}. Then, under Assumption \ref{assumption2}, if $\phi$ is u$\delta$-PE, where $c_1 > N$, the origin of the closed-loop system \eqref{eq1} with $\hat x(t)$ given by the update law \eqref{eq2} is uniformly globally asymptotically stable and the topology of network \eqref{vec_E_multi} is identified by the estimate $\hat w$, that is, $\bar w = \lim_{t \to \infty}\hat w (t)$ with the update law \eqref{updatelaw_two}. Moreover, the system achieves synchronization.
\end{proposition}
\begin{proof}
Consider the Lyapunov function in \eqref{V1}. Its derivative along the trajectories of \eqref{eq1} yields
\begin{align}\label{Vdot}
	\dot V_1 (\tilde x, \tilde w) &= - \tilde{x}^\top (c_1 I + L ) \tilde x \leq -(c_1 + \lambda_{\min}(L)) \lvert \tilde x \rvert^2.
\end{align}

In what follows, we provide a lower bound for $\lambda_{\min}(L)$, namely, $\lambda_{\min}(L)\ge -N$. Using this result, for \eqref{Vdot} to be negative, it is sufficient to choose $c_1 > N$, with the remainder of the proof following from Proposition \ref{prop1}.

To prove that $\lambda_{\min}(L)\ge -N$ holds, we proceed as follows.
Note that, since the signed graph is undirected ($A=A^\top$), we have 
$x^\top L x = \frac{1}{2} \sum_{ij}a_{ij} (x_i -x_j)^2$. Then, under the assumption of normalized weights $|a_{ij}|\le 1$ for all $i,j =1,\dots, N$, and of no self-loops $a_{ii}=0$, we obtain $x^\top L x = \frac{1}{2} \sum_{ij}a_{ij} (x_i -x_j)^2 \ge -\frac{1}{2} \sum_{ij} (x_i -x_j)^2 = - N \sum_{i} x_i^2 + \sum_{i,j} x_i x_j = - N \sum_{i} x_i^2 + (\sum_i x_i)^2 \ge - N \sum_{i} x_i^2 = - N |x|^2.$
% \begin{align*}
%     x^\top &L x 
%     %=\frac{1}{2}\sum_i x_i^2\big(\sum_j a_{ij}) + \frac{1}{2}\sum_j x_j^2\big(\sum_i a_{ji})- \sum_{ij}x_ix_ja_{ij}\\
%     %&= \frac{1}{2} \Big(\sum_{ij}a_{ij} (x_i^2 + x_j^2 -2x_ix_j) \Big)
%     = \frac{1}{2} \sum_{ij}a_{ij} (x_i -x_j)^2
%     \\&\ge -\frac{1}{2} \sum_{ij} (x_i -x_j)^2
%     \\&= - N \sum_{i} x_i^2 + \sum_{i,j} x_i x_j
%     \\&= - N \sum_{i} x_i^2 + (\sum_i x_i)^2
%     \\& \ge - N \sum_{i} x_i^2 = - N |x|^2
%     % = \frac{1}{2} \sum_{ij}a_{ij} (x_i -x_j)^2
%     % \\&\ge -\frac{1}{2} \sum_{ij} (x_i -x_j)^2
%     % \\&= - (N-1) \sum_{i} x_i^2 + \sum_{i,j\ne i} x_i x_j
%     % \\&= - (N-1) \sum_{i} x_i^2 + \big( (\sum_i x_i)^2 - \sum_i x_i^2\big)
%     % \\&= - N \sum_{i} x_i^2 + (\sum_i x_i)^2
%     % \\& \ge - N \sum_{i} x_i^2 = - N |x|^2
% \end{align*}
It follows from Rayleigh's Theorem \cite[Theorem 4.2.2]{Horn} that $\lambda_{\min}(L) = \min_{|x|=1} x^\top L x \ge -N$. Equality corresponds to the \textit{worst case scenario}, where the signed graph is complete and has all negative edges.
\end{proof}

The proofs of Propositions \ref{prop1} and \ref{prop2} rely on the fact that $\phi(t, \xi)$ in \eqref{eq2} is u$\delta$-PE with respect to $\xi$. However, this would mean that the design of $\phi$ would depend on $\tilde w_s$, and consequently on $\bar w_s$, which are the unknown weights to be estimated. Therefore, in the following, we study the stability properties of \eqref{eq1} and \eqref{eq2}, when $\phi$ is u$\delta$-PE only with respect to $\tilde x$. To that end, we recall the following definition and theorem from \cite{chaillet2008uniform}, which will be used in the proof of Proposition \ref{prop4}.
Consider
\begin{align}\label{a}
    \dot{x}=f(t,x,\theta),
\end{align}
where $x \in \mathbb{R}^n, t \in \mathbb{R}_{\geq 0},$ and $\theta \in \mathbb{R}_{n_{\theta}}$ is a constant parameter, and $f(t,x,\theta)$ is piecewise continuous in $t$ and locally Lipschitz in $x$ for all $\theta$.
\begin{definition}(\cite{chaillet2008uniform})
Let $\Theta \in \mathbb{R}_{n_\theta}$ be a set of parameters.  
The system \eqref{a} is said to be \textit{uniformly semiglobally practically asymptotically stable} on $\Theta$ if, given any $\Delta > \delta > 0$, there exists $\theta^*(\delta, \Delta) \in \Theta$ such that $\mathbb{B}(\delta)$ is uniformly asymptotically stable on $\mathbb{B}(\Delta)$ for the system $\dot{x} = f(t, x, \theta^*).$
\end{definition}
The following theorem provides a Lyapunov-based condition for uniform semiglobal practical asymptotic stability of the system \eqref{a}.
\begin{theorem} [\cite{chaillet2008uniform}] \label{th1}
Suppose that given any $\Delta > \delta > 0$, there exist a parameter $\theta^*(\delta, \Delta) \in \Theta$, a continuously differentiable function 
$V_{\delta, \Delta} : \mathbb{R}_{\ge 0} \times \mathbb{R}^n \to \mathbb{R}_{\ge 0}$, class $\mathcal{K}_\infty$ functions 
$\underline{\alpha}_{\delta, \Delta}$, $\overline{\alpha}_{\delta, \Delta}$, and $\alpha_{\delta, \Delta}$ such that for all 
$x \in \mathbb{H}(\delta, \Delta)$ and all $t \in \mathbb{R}_{\ge 0}$, it holds that $\underline{\alpha}_{\delta, \Delta} (|x|_\delta) \le V_{\delta, \Delta}(t, x) \le \overline{\alpha}_{\delta, \Delta} (|x|)$
and
$\frac{\partial V_{\delta, \Delta}}{\partial t}(t, x) 
+ \frac{\partial V_{\delta, \Delta}}{\partial x}(t, x) f(t, x, \theta^*) 
\le - \alpha_{\delta, \Delta} (|x|).$
Furthermore, assume that, given any $\Delta^* > \delta^* > 0$, there exists $\Delta > \delta > 0$ such that $\underline{\alpha}^{-1}_{\delta, \Delta} \circ \overline{\alpha}_{\delta, \Delta}(\delta) \le \delta^*,$ $\overline{\alpha}^{-1}_{\delta, \Delta} \circ \underline{\alpha}_{\delta, \Delta}(\Delta) \ge \Delta^*.$ Then, the system \eqref{a} is uniformly semiglobally practically asymptotically stable on the parameter set $\Theta$.
\end{theorem}
We are now ready to present the following result.

%From the dynamics of \eqref{eq2}, we denote $\hat Z(t)=\hat Z(t,\tilde x)$. 
\begin{proposition}\label{prop4}
Let Assumption \ref{assumption2} hold. Let $\hat z : \mathbb{R}_{\geq 0} \times \mathbb{R}^{\bar M} \times \mathbb{R}^{n_\theta} \to \mathbb{R}^{\bar M},
	(t, \tilde{x}, \theta) \mapsto \hat z(t,\tilde{x},\theta),$ 
    and $\phi : \mathbb{R}_{\geq 0} \times \mathbb{R}^{\bar M} \times \mathbb{R}^{n_\theta} \to \mathbb{R}^{\bar M}, (t, \tilde{x}, \theta) \mapsto \phi(t,\tilde{x},\theta),$
	be parameterized by constants $\theta \in \Theta \subset \mathbb{R}^{n_\theta}$.
Let $\rho : \mathbb{R}_{\geq 0} \to \mathbb{R}_{\geq 0}$ be a continuous non-decreasing function, and assume that  $\hat z_\theta(t,\tilde{x}) := \hat z(t,\tilde{x},\theta)$ satisfies 
	\begin{equation}\label{19}
		\max \left\{ |\hat z_\theta(\tilde{x})|, \left| \frac{\partial \hat z_\theta(\tilde{x})}{\partial t} \right|, \left| \frac{\partial \hat z_\theta(\tilde{x})}{\partial \tilde{x}} \right| \right\} \leq \rho(|\tilde{x}|).
	\end{equation}
    Let the control input $u$ be as in \eqref{u}, where $c_1>N$ and $\hat x(t)$ is given by the update law 
    \begin{align}\label{eq2_theta}
	\dot{\hat{x}}(t) = -\hat{x}(t) + \phi(t,\tilde{x},\theta).
    \end{align}
     %If $\phi_{\theta}(t, \tilde x)$ is chosen such that
     If $\bar E\hat Z_\theta(t, \tilde x)$ is u$\delta$-PE with respect to $\tilde{x}$, where $\hat Z_\theta := \mbox{diag}(\hat z_{\theta})$ and $\hat z_{\theta}=\bar E^\top \hat x$, then the origin of the closed-loop system \eqref{eq1} is \textit{uniformly semiglobally practically asymptotically stable} on $\Theta$.  
\end{proposition} 
\begin{proof}
Consider the following candidate Lyapunov function $V$.
\begin{equation}\label{303}
    \begin{aligned} 
        V(t,\xi)&:=V_1(t, \xi)+\varepsilon V_4(t,\xi)\\
        V_4(t,\xi)&:=V_2(t,\xi)+V_3(t,\xi)\\
	V_2(t,\xi)&:=\tilde{x}^\top \bar E\hat Z_{\theta}(t,\tilde{x}) \tilde{w}_s\\
	 V_3(t,\xi)&:=- \tilde{w}_s^\top \tilde{w}_s\int_{t}^{\infty}e^{(t-\tau)} |\bar E \hat Z_{\theta}(\tau,\tilde{x})|^2 d\tau ,
	     \end{aligned}\end{equation}
where $\xi^\top = [\tilde x^\top \ \tilde{w}_s ^\top]$, $V_1(t,\xi)$ is given in \eqref{V1}, and $\varepsilon>0$. 

We now work toward applying Theorem~\ref{th1}. As a first step, we derive upper and lower bounds for $V$, by providing corresponding bounds for the functions $V_2$, $V_3$, and $V_4$. Since $\bar E\hat Z_{\theta}$ is u$\delta$-PE, for all $(t,\xi)\in\mathbb{R}_{\geq 0}\times\mathbb{H}(\delta, \Delta)$, we get
 \begin{align}
     V_3(t,\xi)%=& - \tilde{w}_s^\top \int_{t}^{\infty}e^{(t-\tau)} \hat Z_{\theta}(\tau,\tilde{x}){\bar E}^\top \bar E\hat Z_{\theta}(\tau,\tilde{x}) d\tau  \tilde{w}_s \notag\\
     \leq& - \tilde{w}_s^\top \tilde{w}_s \int_{t}^{t+T}e^{(t-\tau)}| {\bar E}\hat Z_{\theta}(\tau,\tilde{x}) |^2 d\tau   \notag\\
   %\leq  &-e^{-T} \mu |\tilde{w}_s |^2, 
    \leq  &-e^{-T} \mu |\tilde{w}_s |^2,
\label{v3b}
 \end{align}
where $\mu$ and $T$ are as given in Definition~\ref{delta-PE} for the u$\delta$-PE function $\bar E\hat Z_{\theta}$. Note that the last inequality is obtained using $e^{(t-\tau)}\geq e^{-T}$ for $\tau \in [t,t+T]$, and Lemma~\ref{lemma2}. Denote: $b':=e^{-T} \mu>0, \;\; b:=|\bar E|>0, \;\; \brho:=b\rho(\Delta)\ge 0.$ %b\rho(|\tilde x|).
% b' &:=e^{-T} \mu >0, \;\; b^*:=e^{-T} \gamma_{\delta}(|\tilde x|,\theta) >0, \nonumber\\
% b &:=|\bar E|>0, \;\; \brho:=b\rho(\Delta)\ge 0.}
Based on \eqref{19} and \eqref{v3b}, $ V_4(t,\xi)$ satisfies, for all $(t,\xi)\in\mathbb{R}_{\geq 0}\times\mathbb{H}(\delta, \Delta)$, 
	\begin{equation*}%\label{326}
	V_4(t,\xi)\leq b |\tilde{x} |\rho(|\tilde{x}|)|\tilde{w}_s| - b'|\tilde{w}_s|^2
    \leq \brho |\tilde{x}| |\tilde{w}_s| - b'|\tilde{w}_s|^2.
	\end{equation*} 
 %Define $\brho:=b\rho(\Delta)$.
In addition, for the lower bound of $V_3$, we have $| {\bar E}\hat Z_{\theta}(\tau,\tilde{x})|^2 \leq b^2 \rho^2(|\tilde x|)\leq b^2 \rho^2(|\Delta|)$. Therefore, $\varepsilon V_4(t,\xi)$ satisfies on $\mathbb{R}_{\geq 0}\times\mathbb{H}(\delta, \Delta)$  
	\begin{equation}\label{333}\begin{aligned}
	-\varepsilon \brho |\tilde{x}| |\tilde{w}_s|-\varepsilon \brho^2 |\tilde{w}_s|^2 &\leq \varepsilon V_4(t,\xi)\\&\leq \varepsilon \brho |\tilde{x} | |\tilde{w}_s| - \varepsilon  b' |\tilde{w}_s|^2.
	\end{aligned}\end{equation}
From \eqref{V1} and \eqref{333}, it follows that for  any $\Delta>0$ and for a sufficiently small $\varepsilon$, there exist constant $\underline{\alpha}_{\delta,\Delta}>0$ and $\overline{\alpha}_{\delta,\Delta}>0$ such that for all $(t,x)\in\mathbb{R}_{\geq 0}\times\mathbb{H}(\delta, \Delta)$
	\begin{equation}\label{341}
	\underline{\alpha}_{\delta,\Delta} |\xi |^2\leq V(t,\xi) \leq \overline{\alpha}_{\delta,\Delta}|\xi |^2.
	\end{equation}
Next, we obtain the derivative of $V(t,\xi)$ along the trajectories of the system \eqref{eq1}. Since $\dot V_1(t,\xi)$ is given in \eqref{dotV1}, we focus the analysis on the derivative of $V_4(t,\xi)$, thus on $\dot V_2$ and $\dot V_3$. 
Let $\dot{\hat z}_{\theta}(t,\tilde{x}):=\frac{\partial \hat z_\theta(t,\tilde{x})}{\partial t}+\frac{\partial \hat z_\theta(t,{\tilde{x}})}{\partial \tilde{x}}^\top {\dot {\tilde{x}}} $ and $\dot {\hat Z}_{\theta}=\mathrm{diag}(\dot{\hat z}_{\theta})$. We obtain the following expressions for $\dot V_2$ and $\dot V_3$:
	 \begin{equation*}\label{dotv2}\begin{aligned} 
	\dot V_2(t,\xi)\le &|\hat Z_{\theta}(t,\tilde{x}){\bar E}^\top \tilde{x}|^2 + \tilde{w}_s^\top \hat Z_{\theta}(t,\tilde{x}) {\bar E}^\top (-c_1 I - L)\tilde{x}\\
	&\; - | {\bar E} \hat Z_{\theta}(t,\tilde{x}) \tilde{w}_s|^2 + \tilde{w}_s^\top \dot{\hat Z}_{\theta} (t,\tilde{x}){\bar E}^\top \tilde{x},\\  
     \dot V_3(t,\xi) = & -2\tilde{w}_s^\top \dot {\tilde{w}}_s I(t) - |\tilde w_s|^2 I(t) + |\tilde w_s|^2 | {\bar E} \hat Z_{\theta}(t,\tilde{x})|^2\\
      & - |\tilde w_s|^2 \int_{t}^{\infty}e^{(t-\tau)} \frac{\partial}{\partial t}|\bar E \hat Z_{\theta}(\tau,\tilde{x})|^2 d\tau,
    % \color{blue}\dot V_3(t,\xi) &\color{blue}\le -2\tilde{w}_s^\top \dot {\tilde{w}}_s \int_{t}^{\infty}{e^{(t-\tau)}| {\bar E}\hat Z_{\theta}(\tau,\tilde{x}) }|^2 d\tau\\ 
    % &\color{blue} \!\!-\!\! | \tilde{w}_s|^2 \int_{t}^{\infty} e^{(t-\tau)}|{\bar E}\hat Z_{\theta}(\tau,\tilde{x}) |^2\! d\tau  +\!\left |{\bar E}\hat Z_{\theta}(t,\tilde{x})\!\tilde{w}_s\right|^2.
	\end{aligned}\end{equation*}
    where $I(t) := \int_{t}^{\infty}e^{(t-\tau)} |\bar E \hat Z_{\theta}(\tau,\tilde{x})|^2 d\tau.$
% In addition to the parameters defined in \eqref{proof3:param_b}, we introduce the notations
% \begin{align*}
%         \a_1&:= c_1 + \lambda_{\max} (L),\quad \a_2:= c_1 + \lambda_{\min} (L)\\
%         q(y,s)&:=\big((1+\a_1+\a_1 \Delta)\brho+\Delta\brho^2+2(\brho+\a_1 \Delta)\brho^2\big) y s\\
%         & \ \ \  +\brho^2 y^2+2 \Delta \brho^3 s^2,
%     \end{align*}
    %  ---- original
    % \begin{align*}
    %     \a&:= c_1 + \lambda_{\min} (L)\\
    %     q(y,s)&:=\brho[(1+2\brho^2 +\a+2\brho^2 \a \Delta+\a \Delta +\brho\Delta) y s\\
    %     & \ \ \  +\brho y^2+2\brho^2 \Delta s^2].
    % \end{align*}
Using \eqref{19}, we obtain $\dot {\hat Z}_{\theta}\le \rho(\Delta) +\a_1\rho(\Delta)|\tilde x|+\brho\rho(\Delta) |\tilde{\omega}_s|$.  
    Moreover, for $(t,\xi)\in\mathbb{R}_{\geq 0}\times\mathbb{H}(\delta, \Delta)$, the latter allows us to derive upper bounds for $\dot V_2$ and $\dot V_3$,  
     \begin{align*}
     % V2--
          \dot V_2 \leq &\brho^2 |\tilde x|^2 + \a_1\brho|\tilde x||\tilde w_s| - | {\bar E} \hat Z_{\theta}(t,\tilde{x})  \tilde{w}_s|^
          2\\ & + (\brho+ \a_1\brho\Delta +\brho^2 \Delta)|\tilde x||\tilde w_s|
     \\
     % V3---
      \dot V_3 \leq &2\brho^3 |\tilde x||\tilde{w}_s | + V_3 + |\tilde w_s|^2 | {\bar E} \hat Z_{\theta}(t,\tilde{x})|^2 \\
      &+ 2 \a_1 \brho^2 \Delta |\tilde x||\tilde{w}_s | +2 \brho^3 \Delta |\tilde{w}_s |^2, 
     \end{align*}
     where $\a_1:= c_1 + \lambda_{\max} (L)$ and $\lambda_{\max}(L)$ is the maximum eigenvalue of $L$.
 %     which, in turn, yield
	%  \begin{equation}\label{318}
	% \dot V_4(t,x)\leq q( |\tilde{x} |, |\tilde{w}_s |) - b' |\tilde{w}_s |^2.
	% \end{equation} 
    %The last term in $\dot{V}_3$ is derived from the fact that the second term of $\frac{\partial V_3}{\partial t}$ is $V_3$. 
	Therefore, from \eqref{dotV1}, the derivative of $V(t,\xi)$ satisfies, for all $(t,\xi)\in\mathbb{R}_{\geq 0}\times\mathbb{H}(\delta, \Delta)$, 
	 \begin{equation*}\begin{aligned}[b]
	\dot V(t,\xi) \leq& -\a_2 |\tilde{x} |^2 + \varepsilon (\a_1 \brho+\brho+\a_1\brho \Delta + 2\a_1\brho^2\Delta + \brho^2\Delta \\&
    +2\brho^3) |\tilde{x} | |\tilde{w}_s|  + 2\varepsilon \brho^3 \Delta |\tilde{w}_s |^2 + \varepsilon \brho^2 |\tilde{x} |^2 - \varepsilon b' |\tilde{w}_s|^2,
	\end{aligned}\end{equation*} 
    where $\a_2:= c_1 + \lambda_{\min} (L)$ and $\lambda_{\min}(L)$ is the minimum eigenvalue of $L$.
    Using Young's inequality, we have 
	\begin{equation*}\begin{aligned}[b] 
	\dot V(t,\xi) 	\leq& -(\a_2-\varepsilon(1+5\brho^2+4\brho^2\a_1^2
\Delta^2+\a_1^2 +  \brho^4 ) ) |\tilde{x} |^2\\&+\varepsilon ( \brho^2/2+\brho^4/2+\a_1^2\Delta^2/4+\Delta^2/4-b') |\tilde{w}_s |^2.  
	\end{aligned}\end{equation*}  
%     Choosing $\theta^*(\delta, \Delta)$ and thus $\mu$ and $T$ according to Definition~\ref{delta-PE} such that  $b' \geq \brho^2/2+\brho^4/2+\a^2\Delta^2/4+\Delta^2/4 +\beta'$ and $\beta'>0$ yields
% 	\begin{equation*}\begin{aligned}[b]
% 	\dot V(t,x) 	\leq& -(\a-(1+5\brho^2+4\brho^2\a^2
% \Delta^2+\a^2 +  \brho^4 )\varepsilon   ) |\tilde{x} |^2\\&- \varepsilon \beta' |\tilde{w}_s |^2. 
% 	\end{aligned}\end{equation*}
%Here we can choose  $\mu$, $T$  or $\theta^*(0, \Delta)$ such that $b' \geq \brho^2/2+\brho^4/2+ 2\brho^3 \Delta+\a^2\Delta^2/4+\Delta^2/4 +\beta'$.    
 Choosing $\theta^*(\delta, \Delta)$ and $\varepsilon$ sufficiently small, or
the parameter $\a_2$ sufficiently large, such that $\a_2-\varepsilon\left(1+5\brho^2+4\brho^2\a_1^2\Delta^2+\a_1^2+\brho^4\right)  >\alpha$ and $b' \geq \brho^2/2+\brho^4/2+\a_1^2\Delta^2/4+\Delta^2/4 +\beta'$ with $\beta'>0$,
 yields
  \begin{equation}\begin{aligned}[b] \label{353}
	\dot V(t,\xi)   \leq& -\alpha |\tilde{x} |^2- \varepsilon \beta' |\tilde{w}_s |^2.
	\end{aligned}\end{equation}
Therefore,  from \eqref{341} and \eqref{353}, by Theorem \ref{th1}, the origin of \eqref{eq1} is uniformly semi-globally practically asymptotically stable on $\Theta$.
	\end{proof}

\section{Numerical Example}
We consider a network of $N=12$ agents modeled by \eqref{CN} interconnected over an unknown and signed network, represented in Figure \ref{fig_top}. The objective is to identify the unknown signed weights $\bar w_{k_s}$. The external input $u$ is given by \eqref{u} with \eqref{updatelaw_two} and \eqref{eq2_theta}, where $\phi_{\theta}(t,\tilde x) = \tanh(\kappa \tilde x)p_e(t)$, $p_e(t) = 0.5 \sin(15 \pi t) + 0.3 \cos(6 \pi t) - 0.5 \sin(8 \pi t) +0.7 \cos(12 \pi t) + 2 \sin(\pi t) - 0.3 \cos(2 \pi t) - 0.8 \sin(18 \pi t)$, $c_1=13 >N$ and $\kappa=1000$. The simulation results are presented in Figures~\ref{fig1}--\ref{fig3}. As shown in Proposition \ref{prop4}, the weight estimation errors $\tilde w_{s}$ converge to a small neighborhood of the origin (Figure~\ref{fig1}), indicating that the signed topology is successfully identified (see Figure~\ref{fig2}). Although the estimation errors do not converge exactly to zero, the neighborhood can be made arbitrarily small by appropriately tuning the parameter $\kappa$.  Finally, the proposed topology identification and synchronization strategy operates, with a small tunable error, without requiring global knowledge of the network (Figure~\ref{fig3}).

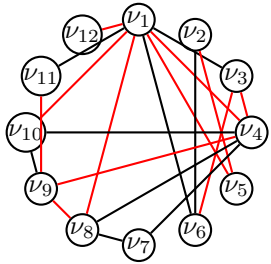
\begin{figure}[!htb]
\centering
\begin{tikzpicture}[thick,scale=0.6,
  mynode/.style={circle,draw,minimum size=12pt,inner sep=0pt,thick},
  coop/.style={- , red, thick},
  comp/.style={- , black, thick}
  ]

% ==== Nodes ====
\foreach \i [count=\n from 1] in {90,60,30,0,-30,-60,-90,-120,-150,180,150,120}{
    \node[mynode] (\n) at ({2.5*cos(\i)},{2.5*sin(\i)}) {$\nu_{\n}$};
}

% ==== Edges derived from L_s ====
\draw[comp] (1)--(3);
\draw[coop] (1)--(4);
\draw[coop] (1)--(5);
\draw[comp] (1)--(6);
\draw[coop] (1)--(8);
\draw[coop] (1)--(10);
\draw[comp] (1)--(11);
\draw[coop] (1)--(12);

\draw[coop] (2)--(5);
\draw[comp] (2)--(6);

\draw[coop] (3)--(4);
\draw[coop] (3)--(6);

\draw[comp] (4)--(7);
\draw[comp] (4)--(8);
\draw[coop] (4)--(9);
\draw[comp] (4)--(10);

\draw[comp] (7)--(8);

\draw[coop] (8)--(9);

\draw[comp] (9)--(10);
\draw[coop] (9)--(11);
\end{tikzpicture}
\caption{A connected signed graph of the 12 agents. The black lines represent cooperative edges, and the red lines represent the antagonistic edges.}
\label{fig_top}
\end{figure}

\begin{figure}[htbp]
\centering

\subfloat[]{
    \includegraphics[width=0.9\columnwidth]{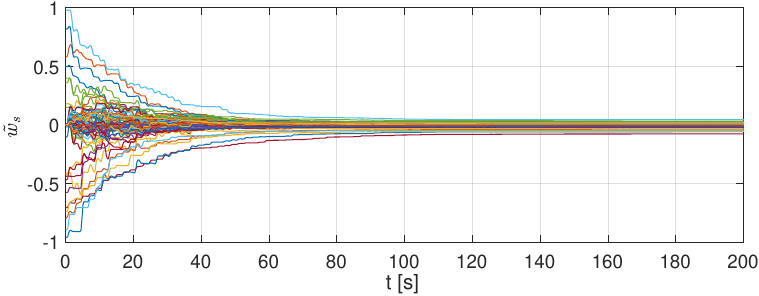}
    \label{fig1}
}
\hfill
\subfloat[]{
    \includegraphics[width=0.9\columnwidth]{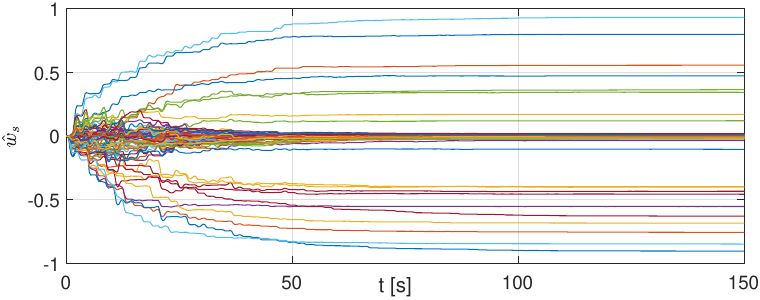}
    \label{fig2}}
    \hfill
\subfloat[]{
    \includegraphics[width=0.9\columnwidth]{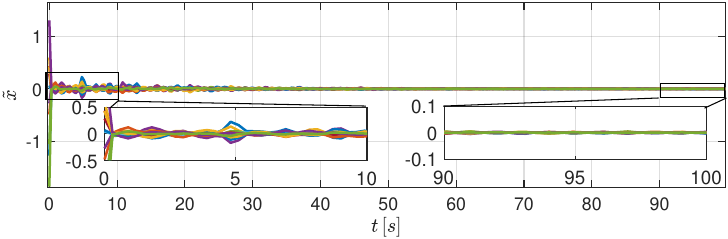}
    \label{fig3}
}

\caption{Evolution of the trajectories. (a): Estimation errors. (b): Estimated weights. (c): Synchronization errors.}
\label{fig:main}
\end{figure}

% \begin{figure}[h!]
% 	\centering
% 	\includegraphics[width=\columnwidth]{Figures/w_tilde.pdf}
% 	\caption{{\small Evolution of the estimation errors.}}
%     \label{fig1}
% \end{figure}

% \begin{figure}[h!]
% 	\centering
% 	\includegraphics[width=\columnwidth]{Figures/w_hat.pdf}
% 	\caption{ {\small Evolution of the estimated weights.}}
%     \label{fig2}
% \end{figure}

% \begin{figure}[h!]
% 	\centering
% 	\includegraphics[width=\columnwidth]{Figures/x_tilde2.pdf}
% 	\caption{ {\small Evolution of the synchronization errors.}}
%     \label{fig4}
% \end{figure}

\section{CONCLUSIONS}
We presented an adaptive control protocol for simultaneous synchronization and identification of the topology of signed dynamical networks with unknown and normalized weights. The design was based on the edge-based formulation, signed Laplacian, and on a known auxiliary network which was designed to be $\delta$-PE. We established uniform semiglobal practical asymptotical stability of the estimation errors. Future work will focus on identifying the topology of signed networks with edge weights of magnitude greater than one, considering discrete-time dynamics, and allowing agents to join or leave the network.

\bibliographystyle{IEEEtran}
\bibliography{references_pelin,reference_nana} %,references_angela

@ARTICLE{Shi2019Dynamics,
author = {Shi, Guodong and Altafini, Claudio and Baras, John S.},
doi = {10.1137/17M1134172},
issn = {0036-1445},
journal = {SIAM Review},
number = {2},
pages = {229--257},
title = {{Dynamics over Signed Networks}},
volume = {61},
year = {2019}
}

@article{Fontan2021Signed,
author = {Fontan, Angela and Altafini, Claudio},
doi = {10.1038/s41598-021-84147-3},
issn = {2045-2322},
journal = {Scientific Reports},
number = {5134},
pmid = {33664333},
publisher = {Nature Publishing Group UK},
title = {{A signed network perspective on the government formation process in parliamentary democracies}},
volume = {11},
year = {2021}
}

@preamble{"\def\nonumero{\def\numerodeitem{}}"}

@string{SIAM = "SIAM J. on Contr. and Opt."}

@preamble{ "\def\nonumero{\def\numerodeitem{}}" }

@ARTICLE{altafini2012_6329411,
	title        = {Consensus Problems on Networks With Antagonistic Interactions},
	author       = {C. Altafini},
	year         = 2013,
	journal      = {IEEE Trans. on Autom. Cont.},
	volume       = 58,
	number       = 4,
	pages        = {935--946},
	doi          = {10.1109/TAC.2012.2224251}
}

@ARTICLE{9965602,
  author={Fontan, Angela and Wang, Lingfei and Hong, Yiguang and Shi, Guodong and Altafini, Claudio},
  journal={IEEE Trans. on Autom. Cont.}, 
  title={Multiagent Consensus Over Time-Invariant and Time-Varying Signed Digraphs via Eventual Positivity}, 
  year={2023},
  volume={68},
  number={9},
  pages={5429-5444}}

@article{chaillet2008uniform,
  title={Uniform semiglobal practical asymptotic stability for non-autonomous cascaded systems and applications},
  author={Chaillet, Antoine and Lor{\'\i}a, Antonio},
  journal={Automatica},
  volume={44},
  number={2},
  pages={337--347},
  year={2008},
  publisher={Elsevier}
}

@article{fontan2023pseudoinverses,
  title={Pseudoinverses of signed Laplacian matrices},
  author={Fontan, Angela and Altafini, Claudio},
  journal={SIAM Journal on Matrix Analysis and Applications},
  volume={44},
  number={2},
  pages={622--647},
  year={2023},
  publisher={SIAM}
}

@inproceedings{restrepo2023simultaneous,
  title={Simultaneous synchronization and topology identification of complex dynamical networks},
  author={Restrepo, Esteban and Wang, Nana and Dimarogonas, Dimos V},
  booktitle={IEEE Conf. on Dec. and Control},
  pages={393--398},
  year={2023}
}

@article{restrepo2023simultaneousTCNS,
  title={Simultaneous topology identification and synchronization of directed dynamical networks},
  author={Restrepo, Esteban and Wang, Nana and Dimarogonas, Dimos V},
  journal={IEEE Trans. on Cont. of Network Sys.},
  volume={11},
  number={3},
  pages={1491--1501},
  year={2023}
}

@inproceedings{loria1999new,
  title={A new notion of persistency-of-excitation for UGAS of NLTV systems: Application to stabilisation of nonholonomic systems},
  author={Lor{\'\i}a, Antonio and Panteley, Elena and Teel, Andrew},
  booktitle={Eur. Control Conf.},
  pages={1363--1368},
  year={1999}
}

@article{loria2005nested,
  title={A nested Matrosov theorem and persistency of excitation for uniform convergence in stable nonautonomous systems},
  author={Lor{\'\i}a, Antonio and Panteley, Elena and Popovic, Dobrivoje and Teel, Andrew R},
  journal={IEEE Trans. on Autom. Cont.},
  volume={50},
  number={2},
  pages={183--198},
  year={2005}
}

@article{panteley2001relaxed,
  title={Relaxed persistency of excitation for uniform asymptotic stability},
  author={Panteley, Elena and Loria, Antonio and Teel, Andrew},
  journal={IEEE Trans. on Autom. Cont.},
  volume={46},
  number={12},
  pages={1874--1886},
  year={2001}
}

@ARTICLE{11045067,
  author={Gogianu, Florin and Buşoniu, Lucian and Morărescu, Irinel-Constantin},
  journal={IEEE Control Systems Letters}, 
  title={Learning-Based Control of the Consensus Value in Unknown Graphs}, 
  year={2025},
  volume={9},
  pages={1093-1098}}

@book{Horn,
  title={Matrix analysis},
  author={Horn, Roger A and Johnson, Charles R},
  year={2012},
  publisher={Cambridge university press}
}

@article{newman2003structure,
  title={The structure and function of complex networks},
  author={Newman, Mark EJ},
  journal={SIAM review},
  volume={45},
  number={2},
  pages={167--256},
  year={2003},
  publisher={SIAM}
}

@article{li2023topology,
  title={Topology inference for network systems: Causality perspective and nonasymptotic performance},
  author={Li, Yushan and He, Jianping and Chen, Cailian and Guan, Xinping},
  journal={IEEE Trans. on Autom. Cont.},
  volume={69},
  number={6},
  pages={3483--3498},
  year={2023}
}

@article{zhu_new_2021,
	title = {A {New} {Method} for {Topology} {Identification} of {Complex} {Dynamical} {Networks}},
	volume = {51},
	doi = {10.1109/TCYB.2019.2894838},
	number = {4},
	journal = {IEEE Trans. on Cybernetics},
	author = {Zhu, Shuaibing and Zhou, Jin and Chen, Guanrong and Lu, Jun-An},
	year = {2021},
	pages = {2224--2231}
}

@article{van_waarde_topology_2019,
	title = {Topology {Reconstruction} of {Dynamical} {Networks} via {Constrained} {Lyapunov} {Equations}},
	volume = {64},
	number = {10},
	urldate = {2025-01-24},
	journal = {IEEE Trans. on Autom. Cont.},
	author = {Van Waarde, Henk J. and Tesi, Pietro and Camlibel, M. Kanat},
	month = oct,
	year = {2019},
	pages = {4300--4306}
}

@article{weerts2018prediction,
  title={Prediction error identification of linear dynamic networks with rank-reduced noise},
  author={Weerts, Harm HM and Van den Hof, Paul MJ and Dankers, Arne G},
  journal={Automatica},
  volume={98},
  pages={256--268},
  year={2018},
  publisher={Elsevier}
}

@inproceedings{wang2023finite,
  title={Finite-time topology identification for complex dynamical networks},
  author={Wang, Nana and Dimarogonas, Dimos V},
  booktitle={62nd IEEE Conf. on Dec. and Cont.},
  pages={425--430},
  year={2023},
}

@article{IEEE-robotmanip,
  title={Robust Formation Control of Robot Manipulators with Inter-agent Constraints over Undirected Signed Networks},
  author={{\c{S}}ekercio{\u{g}}lu, P. and Jayawardhana, B. and Sarras, I. and Lor{\'\i}a, A. and Marzat, J.},
  journal={IEEE Trans. on Cont. of Network Sys.},
    year={2025},
  volume={12},
  number={1},
  pages={251-261}
}
\end{document}